\documentclass[twocolumn,10pt,nofootinbib,superscriptaddress]{revtex4-2}
\usepackage{amsmath,amssymb,amsthm,bm,graphicx,mathtools,braket,dcolumn,multirow}
\usepackage{multirow}
\DeclareSymbolFontAlphabet{\amsmathbb}{AMSb}%

\usepackage[unicode=true,
bookmarks=true,bookmarksopen=false,
breaklinks=false,pdfborder={0 0 0},colorlinks=true]
{hyperref}
\usepackage{xcolor}
\definecolor{cblue}{rgb}{0.16, 0.32, 0.75}
\definecolor{cred}{rgb}{0.7, 0.11, 0.11}
\hypersetup{%
	,linkcolor=cred
	,citecolor=cblue
	,urlcolor=black
}
\usepackage[normalem]{ulem}
\usepackage{bm}

\def\<{\langle}
\def\>{\rangle}
\def\oper{{\mathchoice{\rm 1\mskip-4mu l}{\rm 1\mskip-4mu l}
		{\rm 1\mskip-4.5mu l}{\rm 1\mskip-5mu l}}}

\newcommand{\tr}{\mathrm{Tr}}

\newtheorem{Theorem}{Theorem}
\newtheorem{Proposition}{Proposition}

\newcommand{\hilb}{\mathcal{H}}

\newcommand{\vac}{\mathrm{vac}}
\renewcommand{\i}{\mathrm{i}}
\newcommand{\e}{\mathrm{e}}
\newcommand{\g}{\mathrm{g}}

\begin{document}
	
	\title{
		Quantum regression beyond the Born-Markov approximation\\
		for generalized spin-boson models
	}
	
	\author{Davide Lonigro}
	\email{davide.lonigro@ba.infn.it}
	\affiliation{Dipartimento di Fisica and MECENAS, Universit\`{a} di Bari, I-70126 Bari, Italy}
	\affiliation{INFN, Sezione di Bari, I-70126 Bari, Italy}
	
	\author{Dariusz Chru\'sci\'nski}
	\email{darch@fizyka.umk.pl}
	\affiliation{Institute of Physics, Faculty of Physics, Astronomy and Informatics, Nicolaus Copernicus University, Grudziadzka 5/7, 87-100 Toru\'n, Poland}
	
	\date{\today}
	
	\begin{abstract}
		
		The quantum regression formula for an open quantum system consists in an infinite hierarchy of conditions for its multi-time correlation functions, thus requiring full access to the total ``system+environment'' evolution, and providing a stronger requirement than CP-divisibility. Here, we analyze CP-divisibility and check the validity of quantum regression beyond the Born-Markov approximation (e.g. weak coupling limit) for a class of generalized spin-boson (GSB) models giving rise to a multi-level amplitude-damping evolution; in all cases, it is possible to engineer the system-bath coupling in such a way that quantum regression is exactly satisfied.
	\end{abstract}
	
	\maketitle
	
	\section{Introduction}
	The description of a simple quantum system, for example an atom or molecule, interacting with a thermal environment or an electromagnetic field, is often performed by modeling the former as a few-level system interacting with a structured boson bath; depending on the physical scenario, the energy spectrum of the bath can be either discrete (e.g. a monochromatic laser mode) or continuous.
	
	Spin-boson models, the simplest instances of such models, occupy a central role in the theory of open quantum systems~\cite{Breuer,Legget,Weiss,RIVAS,Ingold}: their simple structure provides a prototypical description of several natural mechanisms like quantum noise, dissipation, and decoherence~\cite{Zoller,Plenio,Car,Zurek,DEC1,DEC2}. For this reason, apart from their fundamental interest, they also find applications in quantum optics~\cite{Zoller,Plenio,Car,Breuer,RIVAS}, quantum information and simulation~\cite{SIM}, and quantum phase transitions~\cite{phase}, to name a few. The interest in such models has also led to an extensive study of their mathematical properties~\cite{davies,fannes,arai,amann,hubner,spohn,hirokawa,hirokawa2}; the wider class of generalized spin-boson (GSB) models, which accommodates multilevel atoms, has also been investigated~\cite{arai2,arai3,takaesu,teranishi,gsb}.
	
	In this article we analyze such models in connection to quantum (non)Markovianity. Quantum non-Markovian dynamics has attracted a lot of attention in recent years (cf.\ recent reviews~\cite{NM1,NM2,NM3,NM4}), and the Markovian property of the evolution induced by spin-boson models has been analyzed by several authors~\cite{Breuer,RIVAS,NM1,NM2,NM3}. Usually, one approaches this problem by studying the properties of the corresponding dynamical map $\Lambda_t$, defined as the reduced evolution of the system (the spin) obtained by tracing out the degrees of freedom of the bath (or environment), represented by a boson field~\cite{Breuer,RIVAS}. This procedure gives rise to a family of maps $\Lambda_t$ such that, for any $t \geq 0$, $\Lambda_t$ is completely positive and trace-preserving (CPTP). The evolution of the system is usually said to be Markovian if $\Lambda_t$ is CP-divisible~\cite{RHP}, that is, for any $t>s$ there exists a CPTP propagator $V_{t,s}$ such that $\Lambda_t = V_{t,s} \Lambda_s$. If $V_{t,s}$ is only positive and trace-preserving, one usually calls $\Lambda_t$ to be P-divisible (cf.~\cite{Sabrina} for a more general classification). A paradigmatic example of CP-divisible evolution is provided by the celebrated Markovian semigroup $\Lambda_t = \e^{t \mathcal{L}}$, with $\mathcal{L}$ being the Gorini-Kossakowski-Lindblad-Sudarshan (GKLS) generator~\cite{GKS,L,Alicki}.
	
	A slightly different concept was proposed in Ref.~\cite{BLP}, where Markovianity is connected to the monotonicity property of the trace norm, that is, for any $t > s$ one requires  $  \| \Lambda_t(\rho_1-\rho_2)\|_1 \leq \| \Lambda_s(\rho_1-\rho_2)\|_1 $ for any pair of initial states $\rho_1$ and $\rho_2$. The violation of this inequality was interpreted in Ref.~\cite{BLP} as the presence of information backflow from the environment to the system, which thus provides a clear indicator of non-Markovian memory effects. The intricate relations between these two notions were studied in Refs.~\cite{Angel,Haikka} and further in Refs.~\cite{BOGNA,PRL-Angel}. Both concepts find a lot of applications in the analysis of several systems (including spin-boson models)~\cite{NM1,NM2,NM3}.
	
	However, it should be stressed that the notion of Markovianity based on the properties of the corresponding dynamical map is not directly related to the original concept of classical Markov stochastic process~\cite{Kampen}. As already pointed out in Refs.~\cite{Hanggi1,Hanggi2}, the divisibility property of the classical dynamical map, represented by a family of stochastic matrices, only provides a necessary condition for the validity of the well-known Markov property for the family of conditional probabilities~\cite{Kampen}:
	\begin{equation}\label{Markov}
		p\!\left(x_n,t_n|x_{n-1},t_{n-1};\ldots;x_1,t_1\right) = p\!\left(x_n,t_n|x_{n-1},t_{n-1}\right).
	\end{equation}
	To provide a quantum generalization of this concept, one therefore needs access to the full system-bath evolution, and not only to the reduced dynamics described by the dynamical map $\Lambda_t$. A mathematical formulation of quantum Markov stochastic processes was proposed in Refs.~\cite{Lewis,Lindblad,Accardi}; it turns out that the quantum Markov property, in analogy to the classical condition~\eqref{Markov}, provides a hierarchy of nontrivial conditions for multi-time correlation functions known as the quantum regression formula~\cite{Lax}. The latter essentially states that the multi-time correlations of operators on the system, derived in terms of the full system-bath evolution, can be recovered in terms of the dynamical map $\Lambda_t$ alone. The quantum regression formula is satisfied whenever the Born-Markov approximation (which consists in assuming that, throughout the evolution, the total ``system+bath'' state has a product form $\rho_{\rm SB}(t) \approx \rho(t) \otimes \rho_{\rm B}$) is valid.
	
	Beyond the Born-Markov approximation, the validity of the quantum regression formula provides highly nontrivial constraints for the very structure of the system-bath interaction. Interestingly, the validity of the quantum regression formula is essentially equivalent to the Markovianity condition proposed in a recent series of papers~\cite{kavan1,kavan2,kavan3} (cf. also the recent review~\cite{Kavan-PRX} and~\cite{NM4} for a comparative analysis). In this approach, Markovianity is shown to correspond to a factorization property of the so-called quantum process tensor of the system. Actually, Ref.~\cite{NM4} provides the whole hierarchy of various diverse concepts of Markovianity and stresses the importance of quantum regression. The analysis of Markovianity based on the quantum regression formula was already initiated in Refs.~\cite{Francesco} and~\cite{Bassano}; these papers show that systems with CP-divisible dynamics may violate quantum regression.
	
	In this article we analyze CP-divisibility and quantum regression for a class of GSB models leading to a multi-level generalization of the familiar amplitude-damping evolution. Such models can describe a multilevel atom with $n\geq1$ excited states and one common ground state, coupled with a structured multi-mode boson field, the energy of each mode having values in a continuous set; the interaction between the atom and the field is modulated by a family of coupling functions. We find sufficient conditions for CP-divisibility to hold and, in particular, a peculiar (and unique) choice of coupling leading to semigroup dynamics at all times is discussed. Besides, in the latter case, the quantum regression formula is shown to hold. This is a highly nontrivial result: contrarily to the standard folklore, in general semigroup dynamics is \textit{not} sufficient to ensure the validity of quantum regression. However, for this class of models, the two notions do coincide.
	
	\section{Spin-boson model and amplitude-damping} 
	Before analyzing in detail the family of GSB models, we shall start by recalling the main features of the paradigmatic spin-boson model leading to amplitude-damping dynamics. Consider a two-level system (qubit) on $\hilb_{\rm S}\simeq\mathbb{C}^2$, with ground and excited state $\ket{\g}$, $\ket{\e}$, interacting with a structured boson field. The energy of the qubit in its excited state will be denoted as $\omega_{\e}$, with its ground energy being set to zero without loss of generality. The total ``system + bath'' Hamiltonian has the following form ($\hbar=1$)
	\begin{equation}\label{eq:H}
		\mathbf{H}=\omega_{\e}\ket{\e}\!\!\bra{\e}\otimes\oper_{\rm B}+\oper_{\rm S}\otimes H_{\rm B} +H_{\rm int},
	\end{equation}
	where $H_{\rm B} = \int \mathrm{d}\omega\;\omega\,b^\dag(\omega) b(\omega)$, and the interaction term reads
	\begin{equation}\label{eq:hint_amplitude}
		H_{\rm int}=\sigma_-\otimes\int \mathrm{d}\omega\; f(\omega)b^\dag(\omega)+\mathrm{h.c.},
	\end{equation}
	where h.c. stands for the Hermitian conjugated term. Above, we use the standard notations $\sigma_-=\ket{\g}\!\!\bra{\e}$, $\sigma_+=(\sigma_-)^\dag=\ket{\e}\!\!\bra{\g}$, and the integral runs over all energies that are accessible to the boson field; $f(\omega)$ is a function (form factor) which modulates the atom-bath coupling, and the bath creation and annihilation operators $b^\dag(\omega)$, $b(\omega)$ satisfy the canonical commutation relations: $[b(\omega),b(\omega')]=0$ and $[b(\omega),b^\dag(\omega')]= \delta(\omega-\omega')$. The vacuum state of the bath will be denoted by $\ket{\vac}$.
	
	The reduced evolution of the system starting from an arbitrary uncorrelated state in the form $\rho\otimes\ket{\vac}\!\bra{\vac}$, given by
	\begin{equation}\label{eq:reduced}
		\Lambda_t(\rho)=\mathrm{Tr}_{\rm B}\left(\e^{-\i t\mathbf{H}}\rho\otimes\ket{\vac}\!\bra{\vac}\e^{\i t\mathbf{H}}\right),
	\end{equation}
	can be computed explicitly. Indeed, the interaction term $H_{\rm int}$ has a rotating-wave form, which implies that the total Hamiltonian preserves the total number of excitations; this conservation law allows one to  compute explicitly the evolution $U_t = \e^{-\i t\mathbf{H}}$ in the single-excitation sector~\cite{Breuer,fl}, that is, if $|\Psi(0)\> = |\e\> \otimes |{\rm vac}\>$, then	
	\begin{equation}\label{}
		|\Psi(t)\> = a(t) |\e\> \otimes |{\rm vac}\> + |\g\> \otimes b^\dagger(\xi_t) |{\rm vac}\> ,
	\end{equation}
	where the function $a(t)$ satisfies the non-local equation	
	\begin{equation}\label{dot-a}
		\i \dot{a}(t) = \omega_\e a(t) + \int_0^t G(t-s) a(s)\;\mathrm{d}s, \;  \ \ a(0)=1
	\end{equation}
	with the memory kernel $G(t) = -\i\int \mathrm{d} \omega\;|f(\omega)|^2 \e^{-\i \omega t}$, and  $ b^\dagger(\xi_t) = \int \mathrm{d}\omega\;\xi_t(\omega) b^\dagger(\omega)$, where the boson profile $\xi_t(\omega)$ is fully characterized by the function $a(t)$ and the form factor $f(\omega)$. Clearly, the state $|\g\> \otimes  |{\rm vac}\> $ does not evolve. Consequently, the reduced dynamics~\eqref{eq:hint_amplitude} reads as follows:
	\begin{equation}\label{eq:amplitude}
		\Lambda_t(\rho)=\begin{pmatrix}
			|a(t)|^2\,\rho_{\e\e} & a(t)\rho_{\e\g}\\
			a(t)^*\rho_{\g\e}& \rho_{\g\g} + (1-|a(t)|^2)\rho_{\e\e}
		\end{pmatrix} .
	\end{equation}
	Eq.~\eqref{eq:amplitude} defines an amplitude-damping channel for any $t>0$; notice that the channel is uniquely characterized by the function $a(t)$ which satisfies $|a(t)|\leq1$. This channel is completely positive and trace-preserving (CPTP) by construction; vice versa, it can be proven (cf. Prop.~\ref{prop:app}) that, for \textit{any} given function $a(t)$, then the following conditions are equivalent:
	\begin{itemize}
		\item $\Lambda_t$ is completely positive;
		\item $\Lambda_t$ is positive;
		\item $|a(t)|\leq1$ for all $t$.
	\end{itemize}
	Therefore, for this model, complete positivity and positivity are equivalent.
	
	\section{Markovianity, divisibility, and semigroup dynamics}
	Recall that the evolution represented by the dynamical map $\Lambda_t$ is divisible if, for any $t > s$, there exists an intermediate map (a propagator) $V_{t,s}$ such that $\Lambda_t = V_{t,s} \Lambda_s$. $\Lambda_t$ is called P-divisible whenever $V_{t,s}$ is positive and trace-preserving, and CP-divisible whenever $V_{t,s}$ is CPTP. The divisibility of the dynamical map is often used to analyze the Markovianity of the corresponding evolution~\cite{NM1}.
	
	Recall that, for an invertible map, P-divisibility (and hence, in our case, CP-divisibility) is equivalent to the monotonicity property~\cite{Angel,Sabrina}
	\begin{equation}\label{<0}
		\frac{\mathrm{d}}{\mathrm{d}t} \| \Lambda_t(X)\|_1 \leq 0 ,
	\end{equation}
	for any Hermitian map in the Hilbert space of the system ($\|X\|_1$ stands for the trace norm). In particular, Eq.~\eqref{<0} implies that, for any pair of initial states $\rho_1$ and $\rho_2$, one has $ \frac{\mathrm{d}}{\mathrm{d}t} \| \Lambda_t(\rho_1-\rho_2)\|_1 \leq 0$. The violation of this inequality was interpreted in Ref.~\cite{BLP} as an information backflow from the environment (bath) to the system, which provides a clear indicator of memory effects.
	In our model, the very condition~\eqref{<0} is equivalent to
	\begin{equation}  \label{MON}
		\frac{\mathrm{d}}{\mathrm{d}t}|a(t)|\leq0.
	\end{equation}
	Therefore, the P-divisibility of $\Lambda_t$ is equivalent to the monotonicity property~\eqref{MON} of $a(t)$. Interestingly, and nontrivially, this condition is indeed necessary and sufficient for CP-divisibility as well:
	\begin{Proposition}\label{prop:div}
		Let $\Lambda_t$ as defined in Eq.~\eqref{eq:amplitude}, with $|a(t)|\neq0$. Then the process is invertible, and the following conditions are equivalent:
		\begin{itemize}
			\item $\Lambda_t$ is CP-divisible;
			\item $\Lambda_t$ is P-divisible;
			\item Eq.~\eqref{MON} holds at all times.
		\end{itemize}
	\end{Proposition}
\begin{proof}
	Particular case of Prop.~\ref{prop:div2}.
\end{proof}
	Therefore Markovianity, when understood as CP-divisibility, reduces to a simple monotonicity condition for $a(t)$.
	
	In particular, one immediately notices that, for the map~\eqref{eq:amplitude}, a semigroup at all times can be only obtained by choosing the parameters in such a way that the function $|a(t)|$ is exponentially decaying at all times. By Eq.~\eqref{dot-a}, this can only happen if the memory kernel reduces to a Dirac delta, which, in turn, requires a \textit{flat coupling} $|f(\omega)|^2=\mathrm{const.}$ on the full real line. If we set
	\begin{equation}   \label{f-AD}
		|f(\omega)|^2=\frac{\gamma}{2\pi},\quad-\infty\leq\omega\leq\infty,
	\end{equation}
	then we get $a(t)=\e^{-\left(\i\tilde{\omega}_{\rm e}+\frac{\gamma}{2}\right)t}$, $|a(t)|^2=\e^{-\gamma t}$, and thus semigroup dynamics	
	\begin{equation}\label{AD-s}
		\dot{\rho}_t = -\i\,\tilde{\omega}_\e [|\e\>\<\e|,\rho_t] + \frac 12  \gamma\!\left(\sigma_- \rho_t \sigma_+ - \frac 12 \{\sigma_+\sigma_-,\rho_t\} \right) .
	\end{equation}
	Here $\tilde{\omega}_{\e}$ is the \textit{renormalized} excitation energy of the qubit, with the bare one being diverging, cf.\ Ref.~\cite{fl}. We stress that no other choice of coupling would yield exact semigroup dynamics at all times, the correspondence between $f(\omega)$ and $a(t)$ being one-to-one; however, relaxing our request by asking the semigroup property to only hold in a finite time window, one would have infinitely many other choices~\cite{hidden1,hidden2}. Also notice that, as a straightforward consequence of the Paley-Wiener theorem, relaxing our request by asking for semigroup dynamics at $t\to\infty$ would still require the energy spectrum of the boson field to be unbounded from below. \cite{khalfin1,khalfin2}
	
	It should be stressed that the semigroup~\eqref{AD-s} is derived \textit{without} any type of Born-Markov approximation. Usually, to derive a Markovian semigroup, one considers two limiting scenarios for some characteristic time scales of the evolution of the system $\tau_{\rm S}$ and bath $\tau_{\rm B}$ \cite{Breuer,Alicki}. In the \textit{weak coupling limit}, one assumes weak coupling between the system and the bath; moreover, $\tau_{\rm B}$ remains constant while $\tau_{\rm S} \to \infty$. This limit works perfectly for many systems in quantum optics~\cite{Zoller,Plenio,Car}. On the other hand, in the \textit{singular coupling limit}, we have $\tau_{\rm B} \to 0$, meaning that the bath correlation functions are delta-like \cite{Hepp}, and this immediately implies memoryless master equations. 
			
	\section{Quantum regression}
	As we have observed, choosing a flat qubit-boson coupling, the dynamical semigroup $\Lambda_t$ is Markovian in the sense of CP-divisibility (or, equivalently, P-divisibility). We will now show that, interestingly, the model also satisfies the more restrictive conditions for Markovianity provided by the quantum regression theorem~\cite{Breuer,Zoller,NM4}.
	
	Given a fixed state of the bath $\rho_{\rm B}$, the unitary evolution of the ``system + bath'' $\mathcal{U}_t \rho = U_t \rho U_t^\dagger$ giving rise to the reduced evolution of the system $\Lambda_t = {\rm Tr}_{\rm B} \mathcal{U}_t(\rho \otimes\rho_{\rm B})$, we say that a pair $(\mathcal{U}_t,\rho_{\rm B})$ satisfies quantum regression if, for any collection of times $t_n > t_{n-1} > \ldots > t_1 > t_0$ and two sets of system operators $\{X_0,X_1,\ldots,X_n\}$ and  $\{Y_0,Y_1,\ldots,Y_n\}$, one has the following relation between multi-time correlation functions:
	\begin{eqnarray}\label{QR}
		&& {\rm Tr}_{\rm SB}[\tilde{\mathcal{E}}_n \mathcal{U}_{t_n-t_{n-1}} \tilde{\mathcal{E}}_{n-1}\mathcal{U}_{t_{n-1}-t_{n-2}}  \ldots \tilde{\mathcal{E}}_0 \mathcal{U}_{t_0}(\rho \otimes \rho_{\rm B})] \nonumber \\
		&=& {\rm Tr}_{\rm S}[\mathcal{E}_n \Lambda_{t_n-t_{n-1}} \mathcal{E}_{n-1}\Lambda_{t_{n-1}-t_{n-2}}  \ldots \mathcal{E}_0 \Lambda_{t_0}(\rho)] ,
	\end{eqnarray}
	where $\tilde{\mathcal{E}}_k A_{\rm SB} = (X_k \otimes \oper_{\rm B}) A_{\rm SB} (Y_k \otimes \oper_{\rm B})$, and  $\mathcal{E}_k A_{\rm S} = X_k  A_{\rm S} Y_k$, for any ``system-bath'' operator $A_{\rm SB}$ and system operator $A_{\rm S}$. Eq.~\eqref{QR} means that all correlation functions for the ``system+bath'' evolution can be computed in terms of the dynamical map of the system alone. Clearly, this is a highly nontrivial infinite hierarchy of conditions, and CP-divisibility provides only a necessary condition for Eq.~\eqref{QR}. It was shown by D\"umcke~\cite{Dumcke} that the quantum regression formula for multi-time correlation functions holds in the weak coupling limit (the validity of Eq.~\eqref{QR} essentially follows from the Born-Markov approximation).
	
	Let us consider now the spin-boson model~\eqref{eq:hint_amplitude} with flat form factor $f(\omega)$, i.e. Eq.~\eqref{f-AD}. As discussed, the reduced evolution defines a dynamical semigroup. However, the validity of the quantum regression formula~\eqref{QR} is not guaranteed since this semigroup were derived without any approximation based on the Born-Markov condition. The first main observation of this article is summarized in the following	
	\begin{Theorem}\label{thm} Let $\rho_{\rm B} = |{\rm vac}\>\<{\rm vac}|$. Then $(\mathcal{U}_t,\rho_{\rm B})$ satisfies the quantum regression hierarchy~\eqref{QR} provided that $|f(\omega)| = {\rm const}.$ for $\omega \in (-\infty,\infty)$.
	\end{Theorem}
	A direct proof of Eq.~\eqref{QR}, which directly shows the cancellation of all terms which might spoil the validity of Eq.~\eqref{QR}, is reported in Appendix~\ref{app:c} for the two-time scenario, the generalization to the $n$-time scenario following immediately from the structure of the model. One may think that quantum regression follows as a direct consequence of the semigroup property, which holds in both cases. This is not true in general: in order to check whether Eq.~\eqref{QR} holds, one needs access to the unitary evolution of the ``system+bath''; only knowing the reduced dynamics is not sufficient.
	
	\section{Multi-level generalization}
	The spin-boson model discussed above is a remarkable example of open quantum system for which necessary and sufficient conditions for CP-divisibility can be found, and, furthermore, admitting a choice of coupling so that the quantum regression formula is \textit{exactly} satisfied. One may think that such a result is exclusive for the qubit case; in fact, similar results can be proven beyond a qubit scenario. We will now discuss a wider class of generalized spin-boson (GSB) models, admitting an arbitrary number of both excited levels of the system and modes of the boson bath, for which similar results can be obtained. This shows the inherently Markovian nature of the amplitude-damping dynamics.
	
	Consider a quantum system living in $\hilb_{\rm S}=\hilb_{\mathrm{e}}\oplus\hilb_{\mathrm{g}}$, with $\dim\hilb_{\mathrm{e}}=n$ and $\dim\hilb_{\mathrm{g}}=1$. $\hilb_{\mathrm{e}}$ corresponds to an $n$-dimensional excited sector, whereas $\hilb_{\mathrm{g}}$ is spanned by the ground state $|\mathrm{g}\>$. The system is coupled to an $r$-mode bosonic bath ($r\leq n$), with the total system-bath Hamiltonian given by ($\hbar=1$)
	\begin{equation}\label{Hn}
		\mathbf{H} = H_{\mathrm{e}} \otimes \oper_{\rm B} + \oper_{\rm S} \otimes \sum_{j=1}^r \int \mathrm{d}\omega\;\omega b^\dagger_j(\omega) b_j(\omega) + H_{\rm int} ,
	\end{equation}
	where $H_{\e}$ is the free Hamiltonian of the excited sector of the system (the energy of the ground state is again set to zero), and the interaction term reads
	\begin{equation}\label{hnint}
		H_{\rm int} = \sum_{j=1}^r\int \mathrm{d}\omega\;f_j(\omega) |\mathrm{g}\>\<\bm{\beta}_j| \otimes b^\dagger_{j}(\omega) + \rm{h.c.},
	\end{equation}
	with $|\bm{\beta}_1\>,\dots,|\bm{\beta}_r\> \in \hilb_{\e}$. This model was first introduced in Ref.~\cite{garrawy}. Similarly as before, each function $f_j(\omega)$ (form factor) modulates the coupling between the system and the $j$th mode of the bath. The bath creation and annihilation operators $b^\dagger_j(\omega)$ and $b_j(\omega)$ satisfy the standard canonical commutation relations: $[b_i(\omega),b_{j}(\omega')]=0$ and $[b_i(\omega),b^\dag_{j}(\omega')]=\delta_{ij}\delta(\omega-\omega')$.
	
	\subsection{Reduced dynamics}
	Again, the single-excitation sector of this model is solvable. This model has a conserved number of excitations, namely
	\begin{equation}
		\mathcal{N}=\oper_{\rm S}\otimes\sum_{j=1}^r\int\mathrm{d}\omega\;b^\dag_{j}(\omega)b_{j}(\omega)+\left(\oper_{\rm S}-\ket{\g}\!\bra{\g}\right)\otimes\oper_{\mathrm{B}},
	\end{equation}
	meaning that $\textbf{H}$ decomposes into the direct sum of its restrictions on the eigenstates of $\mathcal{N}$. In particular, the single-excitation sector, spanned by vectors in the form
	\begin{equation}
		\ket{\Psi}=\ket{\psi_{\e}}\otimes\ket{\vac}+\sum_{j=1}^r\int\mathrm{d}\omega\;\xi_j(\omega)\,\ket{\g}\otimes b^\dag_{j}(\omega)\ket{\vac},
	\end{equation}
	is conserved. This allows us to compute exactly the evolution of any state $\ket{\psi_{\e}}\otimes\ket{\vac}$. The Schr\"odinger equation for the model with this initial condition yields the following linear system:
	\begin{equation}
		\begin{dcases}
			\i\ket{\dot\psi_{\e}(t)}=H_{\e}\ket{\psi_{\e}(t)}+\sum_{\ell=1}^r\int\mathrm{d}\omega\;f_\ell(\omega)^*\xi_\ell(t,\omega)\ket{\bm{\beta}_\ell};\\
			\i\,\dot\xi_j(t,\omega)=\omega\,\xi_j(t,\omega)+f_j(\omega)\braket{\bm{\beta}_j|\psi_{\e}(t)}.
		\end{dcases}
	\end{equation}
	For all $j=1,\dots,r$, we get
	\begin{equation}
		\xi_j(t,\omega)=-\i\int_0^t\mathrm{d}s\,\e^{-\i\omega(t-s)}f_j(\omega)\braket{\bm{\beta}_\ell|\psi_{\e}(s)},
	\end{equation}
	and thus, substituting into the differential equation for $\ket{\psi_{\e}(t)}$, we get
	\begin{eqnarray}
		\i\ket{\dot\psi_{\e}(t)}=
		H_{\e}\ket{\psi_\e(t)}+\int_0^t\mathrm{d}s\;\mathbb{G}(t-s)\ket{\psi_{\e}(s)},
	\end{eqnarray}
	where
	\begin{equation}\label{eq:g}
		\mathbb{G}(t)=-\i\sum_{\ell=1}^r\int\mathrm{d}\omega\;|f_\ell(\omega)|^2\e^{-\i\omega t}\ket{\bm{\beta_\ell}}\!\bra{\bm{\beta}_\ell}.
	\end{equation}
	Summing up, we have shown
	\begin{equation}\label{eq:single1}
		\e^{-\i t\textbf{H}}\ket{\psi_{\e}}\otimes\ket{\vac}=\mathbb{A}(t)\ket{\psi_{\e}}\otimes\ket{\vac}+\ket{\g}\otimes b\left(\xi_j(t,\cdot)\right)\ket{\vac},
	\end{equation}
	with $\mathbb{A}(t)$ being defined via $\ket{\psi_\e(t)}=\mathbb{A}(t)\ket{\psi_\e}$; besides, since $\textbf{H}\ket{\g}\otimes\ket{\vac}=0$, the evolution of the latter state is trivial,
	\begin{equation}\label{eq:single2}
		\e^{-\i t\textbf{H}}\ket{\g}\otimes\ket{\vac}=\ket{\g}\otimes\ket{\vac}.
	\end{equation}
	Finally, defining
	\begin{equation}\label{eq:red}
		\Lambda_t(\rho)=\mathrm{Tr}_{\rm B}\left(\e^{-\i t\mathbf{H}}\rho\otimes\ket{\vac}\!\bra{\vac}\e^{\i t\mathbf{H}}\right),
	\end{equation}
 	and using a natural splitting of the density operator of the system:
	\begin{equation}\label{}
		\rho= \left( \begin{array}{cc} \hat{\rho}_\e & |\mathbf{w}\> \\  \< \mathbf{w}| & \rho_\mathrm{g} \end{array} \right) ,
\end{equation}
	then, substituting Eqs.~\eqref{eq:single1}--\eqref{eq:single2} into Eq.~\eqref{eq:red}, and using the fact that, by unitarity, we must necessarily have
	\begin{equation}\label{eq:normalization}
		\|\ket{\psi_{\e}}\|^2+\sum_{j=1}^r\int\mathrm{d}\omega\;|\xi_j(t,\omega)|^2=1,
	\end{equation}
	we find the following formula for the reduced evolution:
	\begin{equation}\label{MAD}
		\Lambda_t(\rho) = \left( \begin{array}{cc} \mathbb{A}(t)\hat{\rho}_\e \mathbb{A}^\dagger(t) & \mathbb{A}(t)|\mathbf{w}\> \\  \< \mathbf{w}|\mathbb{A}^\dagger(t) & \rho_\mathrm{g}(t) \end{array} \right) ,
	\end{equation}
	with $\rho_\mathrm{g}(t) = \rho_{\g} +  {\rm Tr}\left[ (\oper_\e-\mathbb{A}^\dagger(t)\mathbb{A}(t))\hat{\rho}_\e\right]$. The operator $\mathbb{A}(t) : \mathcal{H}_\e \to \mathcal{H}_\e$ provides a multi-level generalization of the function $a(t)$ which appears in Eq.~\eqref{eq:amplitude}. We point out that other generalizations of the amplitude-damping channel were analyzed in Refs.~\cite{giovannetti,wilde}.
	
	\subsection{CP-divisibility, P-divisibility, and regression}
	Like in the qubit case, the map~\eqref{MAD} is completely positive and trace-preserving (CPTP) and, by construction, the condition $\|\mathbb{A}(t)\|_{\rm op}\leq 1$ holds at all times, where $\|\cdot\|_{\rm op}$ is the operator norm; vice versa, it can be proven that, for \textit{any} operator-valued function $\mathbb{A}(t)$, it is completely positive (or, \textit{equivalently}, positive) by assuming $\|\mathbb{A}(t)\|_{\rm op}\leq1$ at all times:		
\begin{Proposition}\label{prop:app}
	Given $\mathbb{A}(t)$, the following conditions are equivalent:
	\begin{enumerate}
		\item[(i)] $\Lambda_t$ is completely positive;
		\item[(ii)] $\Lambda_t$ is positive;
		\item[(iii)] $\|\mathbb{A}(t)\|_{\rm op}\leq1$.
	\end{enumerate}
\end{Proposition}
\begin{proof} (i)$\implies$(ii) is obvious. Let us prove (iii)$\implies$(i). The condition $\|\mathbb{A}(t)\|_{\rm op}\leq1$ is equivalent to $\oper_{\rm e}-\mathbb{A}^\dag(t)\mathbb{A}(t)$ being a positive operator, implying that its spectral decomposition can be written as
	\begin{equation}
		\oper_{\rm e}-\mathbb{A}^\dag(t)\mathbb{A}(t)=\sum_{\ell=1}^n\ket{\bm{\delta}_\ell(t)}\!\bra{\bm{\delta}_\ell(t)}
	\end{equation}
	for some family of vectors $\ket{\bm{\delta}_\ell(t)}$, $\ell=1,\dots,n$. But then the operators
	\begin{equation}
		\mathbb{K}_0(t)=\begin{pmatrix}
			\mathbb{A}(t)&\ket{\bm{0}}\\
			\bra{\bm{0}}&1
		\end{pmatrix},\qquad	\mathbb{K}_\ell(t)=\begin{pmatrix}
			0&\ket{\bm{0}}\\
			\bra{\bm{\delta_\ell}(t)}&0
		\end{pmatrix},
	\end{equation}
	are easily proven to be Kraus operators of $\Lambda_t$, which is thus completely positive.
	
	Finally, let us prove (ii)$\implies$(iii); equivalently, we will prove that, if (iii) does not hold, then (ii) does not hold. Suppose (iii) to be false; then there exists a vector $\ket{\bm{u}_t}$ such that $\bra{\bm{u}_t}\mathbb{A}^\dag(t)\mathbb{A}(t)\ket{\bm{u}_t}>\braket{\bm{u}_t|\bm{u}_t}$, and thus
	\begin{equation}
		\tr\left([\oper_{\rm e}-\mathbb{A}^\dag(t)\mathbb{A}(t)]\ket{\bm{u}_t}\!\bra{\bm{u}_t}\right)<0.
	\end{equation}
	But then one immediately notices that $\Lambda_t$ maps the positive matrix
	\begin{equation}
		\begin{pmatrix}
			\ket{\bm{u}_t}\!\bra{\bm{u}_t}&\ket{\bm{0}}\\\bra{\bm{0}}&0
		\end{pmatrix}
	\end{equation}
	into a matrix with a negative eigenvalue; therefore, $\Lambda_t$ is not a positive map.
\end{proof}
As a result, a generalization of Prop.~\ref{prop:div} holds:
	\begin{Proposition}\label{prop:div2}
		Let $\Lambda_t$ as defined in Eq.~\eqref{eq:amplitude}, with $\det\mathbb{A}(t)\neq0$. Then the process is invertible, and the following conditions are equivalent:
		\begin{itemize}
			\item[(i)] $\Lambda_t$ is CP-divisible;
			\item[(ii)] $\Lambda_t$ is P-divisible;
			\item[(iii)] For all $t\geq0$, the following condition holds:
			\begin{equation}\label{MON2}
				\frac{\mathrm{d}}{\mathrm{d}t} \| \mathbb{A}(t) \|_{\rm op} \leq 0.
			\end{equation}
		\end{itemize}
	\end{Proposition}
	\begin{proof}
	Suppose $\mathbb{A}(t)$ to be invertible for all $t$, so that $\mathbb{A}(t)^{-1}$ exists. Then a direct calculation shows that, for all $t\geq s\geq0$, $\Lambda_t=V_{t,s}\Lambda_s$, where
		\begin{equation}\label{MAD2}
			V_{t,s}(\rho) = \left( \begin{array}{cc} \mathbb{A}(t,s)\hat{\rho}_\e \mathbb{A}(t,s)^\dag & \mathbb{A}(t,s)|\mathbf{w}\> \\  \< \mathbf{w}|\mathbb{A}(t,s)^\dag  & \rho_\g(t,s) \end{array} \right),
		\end{equation}
where $\mathbb{A}(t,s):=\mathbb{A}(t)\mathbb{A}(s)^{-1}$ and $\rho_g(t,s):=\rho_\g +  {\rm Tr}\left[ \left(\oper_\e-\mathbb{A}(t,s)^\dag\mathbb{A}(t,s)\right)\hat{\rho}_\e\right]$. Therefore, by Prop.~\ref{prop:app}, $V_{t,s}$ is (completely) positive if and only if
	\begin{equation}
		\|\mathbb{A}(t,s)\|_{\rm op}\leq1,
	\end{equation}
	that is, $\|\mathbb{A}(t)\|_{\rm op}\leq\|\mathbb{A}(s)\|_{\rm op}$, i.e. $t\mapsto\|\mathbb{A}(t)\|_{\rm op}$ is monotonically decreasing, which in turn is equivalent to Eq.~\eqref{MON2}.
\end{proof}
	
	Furthermore, a dynamical semigroup can only be obtained provided that the operator $\mathbb{A}(t)$ satisfies the property $\mathbb{A}(t)=\mathbb{A}(t-s)\mathbb{A}(s)$. Again, there is an essentially unique choice of form factors $f_j(\omega)$ yielding semigroup dynamics: choosing $|f_j(\omega)|^2 = 1/2\pi$ (via a renormalization procedure analogous to the one discussed in Ref.~\cite{fl} for the case $n=r=1$), one finds $\mathbb{A}(t) = \e^{-(\i \tilde{H}_\e + \Gamma/2)t}$, where
	\begin{equation}
		\Gamma = \sum_{\ell=1}^r \ket{\bm{\beta}_\ell}\!\bra{\bm{\beta}_\ell}
	\end{equation}
	and $\tilde{H}_\e$ is the renormalized energy operator of the excited sector. We stress that the rank of the dissipation operator $\Gamma$ coincides with the number of modes of the boson bath.
	
	Finally, using the same techniques that we exploited in the qubit case, one proves the following multilevel counterpart of Theorem~\ref{thm}:
	\begin{Theorem}\label{thm2} Let $\rho_{\rm B} = |{\rm vac}\>\<{\rm vac}|$, and $\mathcal{U}_t$ the evolution group generated by the Hamiltonian~\eqref{Hn} with $|f_j(\omega)|^2=\mathrm{const.}$ for $\omega\in(-\infty,\infty)$. Then $\left(\mathcal{U}_t,\rho_{\rm B}\right)$ satisfies the quantum regression hierarchy~\eqref{QR}.
	\end{Theorem}
	Again, quantum regression is therefore satisfied under the same exact assumptions that ensure semigroup dynamics.
	
	\section{Conclusions}
	In this article we have analyzed CP-divisibility and quantum regression for a class of generalized spin-boson (GSB) models, describing an atom interacting with a multi-mode boson field, leading to a class of CPTP maps generalizing the amplitude-damping (AD) qubit channels~\cite{Breuer,Legget,Weiss}. Necessary and sufficient conditions for CP-divisibility have been found; moreover, and most importantly, we have shown that it is possible to engineer these models in such a way that quantum regression is exactly and provably satisfied. This provides, up to our knowledge, a first nontrivial class of open systems satisfying this requirement.
	
	We stress that, while the validity of quantum regression for an open quantum system is in general a much stronger condition than the semigroup property for its reduced dynamics, the two notions happen to coincide for the class of models analyzed here. Interestingly, a similar phenomenon has been recently proven to happen for the dephasing-type spin-boson model yielding phase-damping (PD) qubit channels \cite{dephasing}. Whether such a phenomenon is a more general feature remains an open question. Future works will be devoted to the extension of this analysis to other instances of open quantum systems.
	
	Finally, we point out an interesting proposal for checking Markovianity via the so-called conditional past-future independence~\cite{Budini1,Budini2,Budini3}. Again, conditional past-future independence is not guaranteed beyond the Born-Markov approximation. It would be interesting to test spin-boson models and their generalizations using this approach.\\
	
	\section*{Acknowledgments}
	
	D.L. was partially supported by Istituto Nazionale di Fisica Nucleare (INFN) through the project “QUANTUM” and by the Italian National Group of Mathematical Physics (GNFM-INdAM); he also thanks the Institute of Physics at the Nicolaus University in Toru\'n for its  hospitality. D.C. was supported by the Polish National Science Centre Project No. 2018/30/A/ST2/00837. 
	
	\appendix
	
	\begin{widetext}
		
	\section{Proof of quantum regression}\label{app:c}
		
		In this appendix we will provide a proof, in the two-time scenario, of Theorems~\ref{thm}--\ref{thm2}, that is, we will show that
		
		\begin{equation}\label{}
			{\rm Tr}_{\rm B} \Big\{ \mathcal{U}_{t_1-t_0} \Big(X \otimes \oper_{\rm B}) \mathcal{U}_{t_0}( \mathbf{M} \otimes |{\rm vac}\>\<{\rm vac}|) (Y \otimes \oper_{\rm B})\Big)\Big\} = \Lambda_{t_1-t_0}(X\Lambda_{t_0}(\mathbf{M})Y) ,
		\end{equation}
		for $\mathbf{M} = |\e\>\<\e|,|\e\>\<\g|,|\g\>\<\e|,|\g\>\<\g|$.
		\vspace{-0.5cm}
		
		\subsection{$\mathbf{M} = |\e\>\<\e|$}
		
		One has
		
		\begin{equation}\label{}
			{\rm Tr}_{\rm B} \Big\{\mathcal{U}_{t_1-t_0} \Big(X \otimes \oper_{\rm B}) \mathcal{U}_{t_0}(|\e\>\<\e| \otimes |{\rm vac}\>\<{\rm vac}|) (Y \otimes \oper_{\rm B})\Big)\Big\} = \ell_{\e\e}(t_1,t_0) + \ell_{\e\g}(t_1,t_0) +\ell_{\g\e}(t_1,t_0) + \ell_{\g\g}(t_1,t_0) ,
		\end{equation}
		where
		
		\begin{eqnarray}
			\ell_{\e\e}(t_1,t_0)  &=&  |a(t_0)|^2 {\rm Tr}_{\rm B} \Big\{\mathcal{U}_{t_1-t_0} \Big( X |\e\>\<\e|Y \otimes |{\rm vac}\>\<{\rm vac}| \Big)\Big\} , \\
			\ell_{\e\g}(t_1,t_0)  &=&  a(t_0) {\rm Tr}_{\rm B} \Big\{\mathcal{U}_{t_1-t_0} \Big( X |\e\>\<\g|Y \otimes |{\rm vac}\>\<{\rm vac}|b(\xi^*_{t_0}) \Big)\Big\} , \\
			\ell_{\g\e}(t_1,t_0)  &=&  a^*(t_0) {\rm Tr}_{\rm B} \Big\{\mathcal{U}_{t_1-t_0} \Big( X |\g\>\<\e|Y \otimes b^\dagger(\xi_{t_0})|{\rm vac}\>\<{\rm vac}| \Big)\Big\} , \\
			\ell_{\g\g}(t_1,t_0)  &=&   {\rm Tr}_{\rm B} \Big\{\mathcal{U}_{t_1-t_0} \Big( X |\g\>\<\g|Y \otimes  b^\dagger(\xi_{t_0})|{\rm vac}\>\<{\rm vac}|  b(\xi^*_{t_0}) \Big)\Big\} .
		\end{eqnarray}
		Now,
		
		\begin{equation}\label{}
			\Lambda_{t_1-t_0}(X\Lambda_{t_0}(|\e\>\<\e|)Y) = m_{\e\e}(t_1,t_0) + m_{\g\g}(t_1,t_0) ,
		\end{equation}
		where
		
		\begin{eqnarray}
			m_{\e\e}(t_1,t_0)  &=& |a(t_0)|^2 \Lambda_{t_1-t_0}( X |\e\>\<\e|Y ) , \\
			m_{\g\g}(t_1,t_0)  &=& \|\xi_{t_0}\|^2 \Lambda_{t_1-t_0}( X |\g\>\<\g|Y ) .
		\end{eqnarray}
		It is clear that $\ell_{\e\e}(t_1,t_0) = m_{\e\e}(t_1,t_0)$. We prove that
		
		\begin{equation}\label{}
			\ell_{\e\g}(t_1,t_0) =  \ell_{\g\e}(t_1,t_0) = 0\ , \ \ \ell_{\g\g}(t_1,t_0) = m_{\g\g}(t_1,t_0) .
		\end{equation}
		Let
		
		\begin{equation}\label{XY}
			X = \begin{pmatrix}
				x_\e & x_1 \\
				x_2 & x_g
			\end{pmatrix} \ , \ \ \  Y = \begin{pmatrix}
				y_\e & y_1 \\
				y_2 & y_g
			\end{pmatrix}
		\end{equation}
		Now, using the following relation \cite{hidden2}
		
		\begin{equation}\label{}
			{U}_{t_1 -t_0} \oper_{\rm S} \otimes b^\dagger(\xi_{t_0}) =   \oper_{\rm S} \otimes b^\dagger(\xi_{t_1,t_0}) {U}_{t_1 -t_0} ,
		\end{equation}
		where
		
		\begin{equation}\label{}
			\xi_{t_1,t_0}(\omega) = \e^{-\i(t_1-t_0)\omega} \xi_{t_0}(\omega) ,
		\end{equation}
		one has
		\begin{eqnarray}\label{}
			&  U_{t_1-t_0}  X |\e\> \otimes |{\rm vac}\> =  U_{t_1-t_0} ( x_\e |\e \> + x_2 |\g\>) \otimes   |{\rm vac}\> \nonumber \\
			& = x_\e \Big( a(t_1-t_0) |\e\> \otimes |{\rm vac}\>  + |\g\> \otimes b^\dagger(\xi_{t_1-t_0}) |{\rm vac}\> \Big) + x_2 |\g\> \otimes   |{\rm vac}\> ,
		\end{eqnarray}
		and
		
		\begin{eqnarray}\label{}
			&  U_{t_1-t_0}  X |\g\> \otimes  b^\dagger(\xi_{t_0})|{\rm vac}\> =  U_{t_1-t_0} \oper_{\rm S} \otimes  b^\dagger(\xi_{t_0})  U^\dagger_{t_1-t_0}  U_{t_1-t_0} X |\g\> \otimes  |{\rm vac}\> = \oper_{\rm S} \otimes  b^\dagger(\xi_{t_1,t_0}) U_{t_1-t_0} X |\g\> \otimes  |{\rm vac}\> \nonumber \\
			& = x_1 \Big( a(t_1-t_0) |\e\> \otimes b^\dagger(\xi_{t_1,t_0}) |{\rm vac}\> + |\g\> \otimes b^\dagger(\xi_{t_1,t_0}) b^\dagger(\xi_{t_1-t_0}) |{\rm vac}\> \Big) + x_\g |\g\> \otimes b^\dagger(\xi_{t_1,t_0}) |{\rm vac}\> .
		\end{eqnarray}
		Using the following properties:
		
		\begin{eqnarray}
			\< {\rm vac}|\, b(\xi^*_{t_1,t_0}) b^\dagger(\xi_{t_1-t_0})\,| {\rm vac}\>   &=& 0 , \label{xi1}\\
			\| \, b^\dagger(\xi_{t_1,t_0})  b^\dagger(\xi_{t_1-t_0})|{\rm vac}\> \, \|^2 &=& \|\xi_{t_1-t_0}\|^2 \|\xi_{t_0}\|^2 \label{xi2}
		\end{eqnarray}
		one immediately finds $\ell_{\e\g}(t_1,t_0) =  \ell_{\g\e}(t_1,t_0) = 0$, and $\ell_{\g\g}(t_1,t_0) = m_{\g\g}(t_1,t_0)$.

		\subsection{$\mathbf{M} = |\e\>\<\g|$}
		
		One has
		
		\begin{equation}\label{eg}
			{\rm Tr}_{\rm B} \Big\{\mathcal{U}_{t_1-t_0} \Big(X \otimes \oper_{\rm B}) \mathcal{U}_{t_0}(|\e\>\<\g| \otimes |{\rm vac}\>\<{\rm vac}|) (Y \otimes \oper_{\rm B})\Big)\Big\} =  \ell_{\e\g}(t_1,t_0)  + \ell_{\g\g}(t_1,t_0) ,
		\end{equation}
		where
		
		\begin{eqnarray}
			\ell_{\e\g}(t_1,t_0)  &=&  a(t_0) {\rm Tr}_{\rm B} \Big\{\mathcal{U}_{t_1-t_0} \Big( X |\e\>\<\g|Y \otimes |{\rm vac}\>\<{\rm vac}| \Big)\Big\} , \\
			\ell_{\g\g}(t_1,t_0)  &=&   {\rm Tr}_{\rm B} \Big\{\mathcal{U}_{t_1-t_0} \Big( X |\g\>\<\g|Y \otimes  b^\dagger(\xi_{t_0})|{\rm vac}\>\<{\rm vac}|  \Big)\Big\} .
		\end{eqnarray}
		Now,
		
		\begin{equation}\label{}
			\Lambda_{t_1-t_0}(X\Lambda_{t_0}(|\e\>\<\g|)Y) = a(t_0) \Lambda_{t_1-t_0}(X|\e\>\<\g|Y) =  \ell_{\e\g}(t_1,t_0) ,
		\end{equation}
		hence to prove Eq.~\eqref{eg} one has to show that $  \ell_{\g\g}(t_1,t_0) =0$. Again, using Eqs.~\eqref{xi1}--\eqref{xi2} one finds that indeed $\ell_{\g\g}(t_1,t_0)$ does vanish. The proof for $\mathbf{M} =|\g\>\<\e|$ goes along the same line.
		
		\subsection{$\mathbf{M} = |\g\>\<\g|$}
		
		One has
		
		\begin{eqnarray}\label{}
			& {\rm Tr}_{\rm B} \Big\{\mathcal{U}_{t_1-t_0} \Big(X \otimes \oper_{\rm B}) \mathcal{U}_{t_0}(|\g\>\<\g| \otimes |{\rm vac}\>\<{\rm vac}|) (Y \otimes \oper_{\rm B})\Big) \Big\} \nonumber \\
			& = {\rm Tr}_{\rm B} \Big\{\mathcal{U}_{t_1-t_0} (X|\g\>\<\g|Y \otimes |{\rm vac}\>\<{\rm vac}|) \Big\} = \Lambda_{t_1-t_0}(X\Lambda_{t_0}(|\g\>\<\g|)Y) .
		\end{eqnarray}
		The proof for the multi-time scenario goes the same scheme. One can easily observe that, for the multi-level generalization, the technique is analogous but the algebra is of course more tedious.
		
	\end{widetext}

\end{document}